\documentclass[12pt]{article}
\usepackage{mathpazo}
\usepackage[pagewise]{lineno}
\usepackage{amssymb}
\usepackage{graphicx}
\usepackage{amsfonts}
\usepackage{amsmath}
\usepackage{amsthm}
\usepackage{mathtools}
\usepackage{ascmac}
\usepackage{bm}
\usepackage{color}
\usepackage{pgfplots}
\pgfplotsset{compat=1.16}
\usepackage{here}
\usepackage{multirow}
\usepackage[hang,small,bf]{caption}
\usepackage[subrefformat=parens]{subcaption}
\usepackage{authblk}
\usepackage[colorlinks=true, citecolor=teal, linkcolor=teal, urlcolor=teal]{hyperref}
\usepackage[authoryear]{natbib}
\usepackage{threeparttable}
\usepackage{comment}
\usepackage{stackengine}
\usepackage{tikz}
\usetikzlibrary{positioning, calc}
\usetikzlibrary{arrows.meta}
\usetikzlibrary{bending}
\usepackage{setspace}

\def\E#1{\mathbb{E}\left[#1\right]}

\def\P#1{\mathbb{P}\left[#1\right]}

\newcommand{\indep}{\perp\!\!\!\!\perp}

\numberwithin{equation}{section}

\newtheorem{assumption}{Assumption}[section]

\newtheorem{corx}{Corollary}
\newenvironment{corollary}
  {\pushQED{\qed}\corx}
  {\popQED\endcorx}

\newtheorem{example}{Example}

\newtheorem{propx}{Proposition}
\newenvironment{prop}
  {\pushQED{\qed}\propx}
  {\popQED\endpropx}
  
\theoremstyle{definition}
\newtheorem{remarkx}{Remark}
 \newenvironment{remark}
  {\pushQED{\qed}\remarkx}
  {\popQED\endremarkx}

\makeatletter
\renewenvironment{proof}[1][\proofname]{%
  \par\pushQED{\qed}\normalfont%
  \topsep6\p@\@plus6\p@\relax
  \trivlist\item[\hskip\labelsep\bfseries#1\@addpunct{.}]%
  \ignorespaces
}{%
  \popQED\endtrivlist\@endpefalse
}
\makeatother

\makeatletter
\def\blfootnote#1{%
  \begingroup
  \renewcommand\thefootnote{}\footnote{#1}%
  \addtocounter{footnote}{-1}%
  \endgroup
}
\makeatother

\makeatletter
\renewcommand*{\@fnsymbol}[1]{\ensuremath{\ifcase#1\or \flat\or * \else\@ctrerr\fi}}
\makeatother

\renewcommand{\footnotesize}{\fontsize{10pt}{20pt}\selectfont}

\title{\bf\fontsize{18.3pt}{22pt}\selectfont Evaluating Threshold Policies in Randomized Threshold Designs: A Cautionary Tale for Regression Discontinuity Designs}
\author[$\natural$]{Shunsuke Imai}
\author[$\sharp$]{Koshi Nishida}
\author[$\flat$]{Yuta Okamoto}
\affil[$\natural$]{Graduate School of Economics, Kyoto University}
\affil[$\sharp$]{Graduate School of Economics, The University of Tokyo}
\affil[$\flat$]{Graduate School of Economics, Hitotsubashi University}
\date{\today}

\begin{document}
\maketitle
\blfootnote{${}^{\natural}$\href{mailto:imai.shunsuke.57n@st.kyoto-u.ac.jp}{imai.shunsuke.57n@st.kyoto-u.ac.jp}. ${}^{\sharp}$\href{mailto:nishida-koshi841@g.ecc.u-tokyo.ac.jp}{nishida-koshi841@g.ecc.u-tokyo.ac.jp}. ${}^{\flat}$\href{mailto:yuta.okamoto@r.hit-u.ac.jp}{yuta.okamoto@r.hit-u.ac.jp}. 
We gratefully acknowledge Kentaro Kawato, Tomoki Nishiyama, Ryo Okui, and Shosei Sakaguchi for their comments.}
\vspace{-1cm}
\begin{abstract}
\onehalfspacing
    Many policies assign treatment according to whether a score crosses a cutoff. Regression discontinuity (RD) designs identify the effect of treatment assignment, but the threshold policy itself may also reshape individuals’ incentives, inducing behavioral responses that affect outcomes even holding treatment status fixed---a channel that conventional RD designs cannot capture. We develop a framework that exploits randomized variation in policy thresholds, together with rank-invariance-type restrictions, to identify the treatment-assignment and incentive-response effects. We illustrate the empirical relevance of this distinction using data from a merit-based scholarship experiment in Malawi. In this illustration, the conventional RD estimate is positive, while the incentive-response effect is negative and more than twice as large in magnitude. These findings suggest that threshold policies may generate unintended adverse behavioral responses, potentially reflecting discouragement induced by demanding thresholds, and caution against relying solely on conventional RD estimates when evaluating threshold policies.
\end{abstract}

\vspace{0.2cm}


\vspace{0.2cm}


\newpage
\doublespacing

\section{Introduction}\label{sec:intro}

Threshold policies, in which treatment status depends on whether an individual's score crosses a cutoff, are ubiquitous. Examples include test-score thresholds for college admission and financial aid, legal limits on blood alcohol concentration or driving speed, credit-score cutoffs for lending, and emissions thresholds for subsidies or regulatory relief.

Given their ubiquity, understanding the causal effects of threshold policies is central to policy design. Economists have long used regression discontinuity (RD) designs to identify effects operating through treatment assignment at a cutoff. Such estimates provide important evidence on the effectiveness of the treatment itself. However, treatment assignment is not the only channel through which a threshold policy can affect outcomes.

In particular, threshold policies may also alter individuals' incentives regardless of the realized treatment status. In response to the policy, individuals may adjust their effort or other decisions, thereby changing their scores---the running variable in an RD design.\footnote{(Partial) manipulation of the score does not, by itself, invalidate a conventional RD design, provided that agents do not have precise or complete control over the realized score (\citealp{Lee:2008, McCrary:2008}). Hence, the setup considered in this paper does not conflict with the conventional RD framework. Further discussion is provided later.} 
These changes in incentives, behavior, and scores may themselves affect subsequent outcomes, even holding treatment status fixed. For example, a more stringent admissions threshold at a selective university may induce students to study harder, leading to higher test scores and higher future earnings through greater human-capital accumulation.
In this case, even if the conventional RD effect is small or even zero, the threshold policy may still be valuable as an incentive scheme: its value lies not in changing who receives treatment, but in altering the incentives individuals face under the policy.

Despite the empirical relevance of this second channel, most of the RD literature abstracts from it or does not explicitly account for it. At the same time, much of the experimental literature on threshold-based incentives focuses on the effect of assignment to the incentive scheme itself, without separately identifying the effect of reward receipt and the behavioral response induced by the threshold policy. This paper bridges these two approaches by developing a framework that identifies both channels in a directly comparable way: the effect operating through treatment receipt at the cutoff and the effect generated by behavioral responses to the threshold. This decomposition clarifies how the conventional RD estimand relates to the broader policy effect and sharpens the resulting policy implications.

To study these effects, we consider an experimental setting in which individuals are randomly assigned one of two threshold values, $l$ or $h$, with $l<h$. This framework includes a comparison between a specific threshold $l$ and a ``no-policy'' baseline, obtained by setting $h=\infty$. One example of such a design is the following:

\begin{example}[Student Achievement and Retention Demonstration Project]\label{ex1}
    In the Canadian Student Achievement and Retention (STAR) Demonstration Project, studied by \cite{Angrist09}, eligible first-year university students were randomly assigned to a merit-scholarship offer or to a control condition. The offer entitled students to a scholarship if they attained a pre-specified first-year GPA target. Viewing this target as $l$, the experiment corresponds to a randomized comparison between a threshold policy at $l$ and a no-threshold policy, $h=\infty$.
\end{example}

Similar randomized merit-threshold designs are employed in several studies (e.g., \citealp{Kremer09, Leuven10, Berry22}).
Quasi-random variation can also arise when eligibility thresholds vary across cohorts. For example, Colombia's Ser Pilo Paga program uses cohort-specific test-score cutoffs to determine eligibility for financial aid \citep{Londono-Velez_etal:2020aejep,Londono-Velez_etal:2026}. If the distributions of potential outcomes and other relevant characteristics are stable across adjacent cohorts, such threshold changes can be viewed as approximately random.

As illustrated by the studies discussed above, randomized threshold policies have been used in a range of experimental settings. To our knowledge, however, this variation has not been used to separately identify and compare the treatment-assignment and incentive-response effects. We combine randomized variation in threshold regimes with the discontinuous treatment rule within each regime to distinguish the two channels. Even with randomized thresholds, however, the incentive-response effect is not identified without additional restrictions. We show that under a rank-invariance-type assumption, both effects are point identified at the policy cutoff and can therefore be directly compared; we also derive partial-identification results under weaker restrictions. The treatment-assignment effect at the cutoff coincides with the conventional RD estimand.

By disentangling these two effects, our framework allows for more nuanced policy implications. 
For example, if the treatment-assignment effect is positive while the incentive-response effect is negative—as may occur when demanding cutoffs induce discouragement—a need-based treatment rule may be preferable to a merit-based rule because the former may avoid distorting individuals' behavior. By contrast, if the treatment-assignment effect is negligible but the incentive-response effect is positive and sufficiently large, the threshold-based treatment rule may still be justified as an incentive scheme, even though the conventional treatment effect is very small.
This highlights that understanding the effectiveness and consequences of threshold policies requires accounting for both mechanisms.

We apply our identification results to an experimental setting in which students in the treatment group receive merit-based financial awards, whereas those in the control group do not. We examine the effects of this program on students' educational attainment.
In this application, the conventional RD estimate is positive, suggesting a positive effect of award receipt. Taken at face value, this finding might suggest that the program should be implemented. However, we find that the incentive-response effect is negative and more than twice as large in magnitude as the RD estimate. As a result, the overall effect of the policy on students' educational attainment is negative.
These findings suggest that, first, the program may not be desirable in its current form. Second, although financial aid itself may be beneficial, assigning it on the basis of merit may generate adverse incentive-response effects. 

Beyond this particular application, these results yield broader insights for policy design. Our framework captures effects that would be overlooked by a standard RD design alone, thereby enabling a richer economic analysis and more nuanced policy implications.

\paragraph{Plan of the Article.}
This article proceeds as follows. After discussing the related literature in the remainder of this section, Section \ref{sec:motiv} presents a motivating empirical example that illustrates the practical importance of disentangling the two effects. Section \ref{sec:sharp} introduces the setup and target parameters, and then establishes the corresponding identification results. Section \ref{sec:emp} demonstrates the empirical relevance of our framework by applying our procedure to real-world data. 
Section \ref{sec:concl} concludes. All proofs are collected in Appendix \ref{sec:proof} of the Online Appendix.

\paragraph{Related Literature.}

This paper is closely related to the RD literature, which typically exploits threshold-based treatment assignment to identify treatment effects at a cutoff; see \cite{Cattaneo_Titiunik:2022} for a review. Several studies instead consider policy changes in the threshold position. \cite{Dong_Lewbel:2015} study marginal changes in the cutoff, \cite{Bertanha:2020} evaluates counterfactual treatment assignment rules, and \cite{Crippa:2025} studies welfare-maximizing thresholds. These approaches, however, abstract from endogenous responses of the running variable to the policy rule. Our framework instead allows both the score and subsequent outcomes to respond to the threshold policy, which matters because individuals attaining the same observed score under different policy regimes may differ in unobserved characteristics.

Most closely related is the contemporaneous work of \cite{Liao:2025}, which also studies threshold policies when the running variable responds endogenously. Our analysis differs in three main respects. First, while \cite{Liao:2025} focuses primarily on aggregate policy effects under counterfactual cutoffs, our main objects are the treatment-assignment and incentive-response effects themselves. Identifying both at a common policy cutoff allows us to compare their signs and magnitudes directly and to assess how informative the conventional RD estimand is about the broader consequences of a threshold policy. Accordingly, our theoretical and empirical analyses emphasize that the conventional RD estimand can provide a misleading picture when the incentive-response channel is quantitatively important, a question that is not examined by \cite{Liao:2025}. Second, the research designs differ: \cite{Liao:2025} uses a two-period, two-group setting with conditional parallel-trends-type restrictions, whereas we exploit cross-sectional random or quasi-random variation in threshold regimes together with restrictions on cross-regime ranks. Third, our potential outcomes are explicitly indexed by the threshold regime, allowing the threshold itself to affect outcomes through channels other than the score. Our definition of the total policy effect also follows the treatment status induced under each regime, whereas \cite{Liao:2025} compares a treated outcome under the threshold policy with an untreated outcome under the no-policy regime. That said, this difference in formulation appears largely extendable, and the two approaches may be viewed as complementary: when randomized threshold variation is unavailable but a suitable difference-in-differences-type setting exists, \citeauthor{Liao:2025}'s approach may provide an alternative route to evaluating threshold policies with endogenous behavioral responses.

Our work is also related to multi-cutoff RD designs, including \cite{Cattaneo_etal:2021} and \cite{Okamoto_Ozaki:2025}. These studies identify conventional treatment effects away from the cutoff while holding the policy environment fixed, whereas we study causal effects across different threshold policies.

Finally, our framework is related to causal mediation analysis. Rather than relying on standard sequential-ignorability conditions (\citealp{Imai10}; see also \citealp{Heckman_Pinto:2015}), we exploit the deterministic treatment-assignment structure of RD designs together with rank-invariance-type restrictions. This allows us to identify an incentive-response effect conditional on the running variable that is directly comparable to the conventional RD estimand.

\section{Motivating Empirical Example}\label{sec:motiv}

This section uses a financial rewards experiment as a motivating example to illustrate the potential importance of distinguishing between the two effects of a threshold policy. We revisit this setting in Section \ref{sec:emp}, where we apply our proposed procedure.

\cite{Berry22} conducted a randomized controlled trial in Malawi in 2015 to examine the effects of scholarship programs on students' educational attainment. In February 2015, students were randomly assigned to one of three groups: two scholarship programs or a control group. Our analysis focuses on the scholarship program referred to by the authors as the population-based scholarship arm, as well as the control group.
Students assigned to the population-based scholarship arm were eligible to receive an award of MWK 4,500 (USD 9.70) if they scored in the top 15 percent on a final examination administered in June 2015. As noted by \cite{Berry22}, relative to Malawi's per-capita GDP of approximately USD 363, this amount constituted a nontrivial financial incentive. The awards were disbursed in October 2015. We examine the effect of this scholarship program on test scores in the subsequent grade, measured in March 2016.

Two channels are economically relevant in this setting. The first is the effect of receiving the financial award itself. Through the relaxation of household financial constraints, the cash transfer may improve subsequent educational achievement (e.g., by enabling households to purchase educational inputs such as textbooks, improve children's nutritional status, or reduce the need for child labor after the experimental period).

Conversely, the demanding cutoff---requiring students to rank in the top 15 percent---may discourage students for whom the target appears difficult to attain. When the perceived probability of earning the award is low, students may have little incentive to exert additional effort, while the introduction of an external reward may also crowd out intrinsic study motivation, two mechanisms highlighted by \cite{Leuven10}.
This reduction in effort may, in turn, lead to poorer performance on subsequent tests.

The relative importance of these two channels has distinct policy implications. If the treatment-receipt channel dominates, expanding financial assistance may be an effective means of improving educational attainment in Malawi. If the incentive channel is positive and sufficiently large, however, unconditional aid may have limited effects even when the effect of receiving the award is positive; in that case, the design of incentives may be central. Conversely, if the incentive channel is negative and dominates the treatment-receipt channel, a threshold-based award program may reduce educational attainment overall despite a positive effect of receiving the award. In that case, unconditional aid may be preferable.

Thus, the substantive conclusion may differ sharply depending on the sign and relative magnitude of the two channels. This motivates the subsequent analysis.


\section{Randomized Threshold Policy Designs}\label{sec:sharp}

\subsection{Setup}\label{subsec:setup}
Let $C_i\in\{l,h\}$ with $l<h$ denote the policy threshold, $X_i\in \mathcal X \subseteq\mathbb{R}$ the score variable, $Y_i\in\mathbb{R}$ the outcome variable, and $D_i\in\{0,1\}$ the treatment indicator, which equals one if individual $i$ receives the treatment and zero otherwise.
Let $X_i(c)$, $c\in\{l,h\}$, denote the potential score variable that would be realized if the threshold were exogenously set to $c$. 
Throughout this article, we assume that $X_i(c)$ is continuously distributed, as in the RD literature.
We observe $X_i=X_i(h)\mathbf{1}\{C_i=h\}+X_i(l)\mathbf{1}\{C_i=l\}$. The potential treatment status under threshold $c$ is therefore $D_i(c)=\mathbf{1}\{X_i(c)\geq c\}$, and the observed treatment status is $D_i=D_i(C_i)=\mathbf{1}\{X_i(C_i)\geq C_i\}$. 
We can now define the potential outcomes by $Y_i(c,x,d)$ for $c\in\{l,h\}, d\in\{0,1\}$, and $x\in\mathcal X$, and we observe $Y_i = Y_i(C_i, X_i(C_i), D_i(C_i))$.

\begin{figure}
    \centering
    \includegraphics[width=0.5\linewidth]{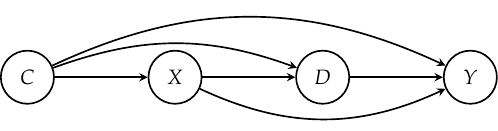}
    \caption{Causal Paths}
    \label{fig:dag1}
\end{figure}
This framework accommodates a variety of causal mechanisms. A schematic representation of the corresponding causal pathways is presented in Figure~\ref{fig:dag1}.
In particular, we allow for both the paths $C\to Y$ and $C\to X\to Y$, which may be economically important in many applications. For example, in the educational setting discussed in Section~\ref{sec:motiv}, the position of the threshold may affect students' motivation to study. Such changes in motivation may directly affect subsequent educational outcomes, corresponding to the path $C\to Y$. They may also alter students' study effort and thereby affect a short-term test score $X$, which may in turn influence the subsequent outcome $Y$, corresponding to the path $C\to X\to Y$.
While conventional RD analyses typically abstract from these pathways, our framework allows them explicitly, and our interest lies in identifying both these effects ($C\to Y$ and $C\to X\to Y$) and the conventional RD effect ($D\to Y$) so that the two can be directly compared.

\begin{remark}[A word of caution]\label{rem:manip}
    To avoid confusion, we here re-emphasize that \textit{partial} score manipulation, whereby the score ``is influenced partially by (1) the individual's attributes and actions, and (2) by random chance'' \citep{Lee:2008}, does not invalidate identification in standard RD designs. 
    The relevant exception is complete manipulation, in which agents have precise and full control over their realized scores. This distinction is apparent in a common RD application in which the running variable is a student’s test score and the treatment is college enrollment. Although students can influence their test scores by adjusting their study effort, they generally cannot determine their realized scores perfectly because performance also depends on idiosyncratic factors. Thus, allowing agents to influence the score through their behavior does not place our framework in conflict with the conventional RD setup.
\end{remark}

Before proceeding to our main identification analysis, we introduce several assumptions that are maintained throughout the analysis.
\begin{assumption}\label{assumption:rct}
    $C_i \indep(X_i(l), X_i(h), \{Y_i(c,x,d)\}_{c\in\{l,h\}, x\in\mathcal{X},d\in\{0,1\}})$.
\end{assumption}

\begin{assumption}\label{assumption:cont}
    For each $c\in\{l,h\}$ and $d\in\{0,1\}$, the following conditions hold:
    \begin{description}
        \item[(i)] $\E{Y_i(c,x,d) \mid X_i(c) = x}$ is continuous in $x \in \mathcal X$.
        \item[(ii)] The support $\mathcal X$ is a compact interval, and $X_i(c)$ has a continuous and strictly positive density on $\mathrm{int} (\mathcal X)$.
    \end{description}
\end{assumption}

\begin{assumption}\label{assumption:propensity}
    $\P{D_i(l) = 1} > \P{D_i(h) = 1}$.
\end{assumption}
\begin{remark}
    This assumption is imposed primarily to simplify the exposition. 
    If the treatment-propensity ordering is reversed, analogous identification and decomposition results can be obtained by interchanging the roles of the treated and untreated states. In the main text, we focus on this case as it appears to be the most empirically relevant settings.    
\end{remark}

Assumption \ref{assumption:rct} holds in the experimental setting where the threshold position is randomly assigned. 
Assumption \ref{assumption:cont} is a mild continuity assumption.
Assumption \ref{assumption:propensity} requires the population treatment propensity to be strictly higher under the more lenient threshold.
In an important special case with $h=\infty$, so that $\P{D_i(h)=1}=0$, the assumption is automatically satisfied provided that $\P{D_i(l) = 1}>0$.

\subsection{Target Parameters}\label{subsec:target}

Threshold policies may affect outcomes through distinct channels.
The first contrast we consider holds the score and threshold fixed and changes only treatment status:
\begin{align*}
    Y_i(c,x,1)-Y_i(c,x,0).
\end{align*}
This contrast captures the individual \textit{treatment-assignment effect} for individuals whose potential score under threshold regime $c$ is equal to $x$.
Define the local average treatment-assignment effect at $x$ by
\begin{align*}
    \tau_c(x)
    =
    \E{Y_i(c,x,1)-Y_i(c,x,0)\mid X_i(c)=x},
    \quad c\in\{l,h\}.
\end{align*}
When $x$ is set equal to the cutoff $c$, $\tau_c(c)$ coincides with the conventional RD estimand and is identified under Assumption~\ref{assumption:rct} and the usual RD continuity condition (\citealp{Hahn_etal:2001}). A formal identification result is provided below.

Now consider a policy change that lowers the threshold from $h$ to $l$.
When $h=\infty$, this policy change corresponds to introducing a threshold policy with cutoff $l$.
In addition to the treatment-assignment effect defined above, we can consider the following contrast:
\begin{align}
    Y_i(l,X_i(l),d)-Y_i(h,X_i(h),d).\label{eq:ind-incentive-effect}
\end{align}
This contrast holds treatment status fixed while changing the threshold regime and the score induced under that regime. It therefore captures the effect of the policy change operating through channels other than treatment assignment, including both changes in the score and other threshold-induced responses that may directly affect the outcome. 

To better understand this effect, it is instructive to decompose 
\begin{align*}
    &Y_i(l,X_i(l),d)-Y_i(h,X_i(h),d)\\
    &= 
    \underbrace{Y_i(h,X_i(l),d)-Y_i(h,X_i(h),d)}_{C\to X\to Y}
    + 
    \underbrace{Y_i(l,X_i(l),d)-Y_i(h,X_i(l),d)}_{C\to Y}.
\end{align*}
The first term captures the effect operating through the change in the score induced by the threshold policy, holding the threshold regime fixed at $h$. The second term captures the remaining effect of changing the threshold regime while holding the score fixed at $X_i(l)$. 
In the context of our empirical application, for example, a threshold policy may alter students' motivation to study, thereby changing their study effort and short-term academic achievement $X$, which in turn affects subsequent academic achievement $Y$. At the same time, the policy-induced change in motivation may persist and directly affect subsequent academic achievement even conditional on the realized score $X$. The two terms above capture these two mechanisms, respectively. In this sense, we refer to their combined effect in \eqref{eq:ind-incentive-effect} as the individual \textit{incentive-response effect}, evaluated at treatment status $d$.

We now define the local average incentive-response effect as follows:
\begin{align*}
    \theta_{c,d}(x) = \E{Y_i(l,X_i(l),d)-Y_i(h,X_i(h),d)\mid X_i(c) = x},\quad \text{for }c\in\{l,h\} \text{ and } d\in\{0,1\}.
\end{align*}
The central challenge in studying threshold policies is the identification of this effect. In general, it is not identified even in an ideal experimental setting. To illustrate, consider individuals whose potential score under threshold $l$ equals the cutoff, $X_i(l)=l$. For these individuals, the average incentive-response effect when treatment status is fixed at $d=0$ is
\begin{align}
    \theta_{l,0}(l)&=\E{Y_i(l,X_i(l),0)-Y_i(h,X_i(h),0)\mid X_i(l)=l}\notag\\
    &=
    \lim_{x\uparrow l}\E{Y_i\mid X_i=x, C_i=l} - \E{Y_i(h,X_i(h),0)\mid X_i(l)=l},\label{eq:second-l}
\end{align}
where the second equality uses Assumptions \ref{assumption:rct}--\ref{assumption:cont}.
The first term is identified from the observed data, whereas the second term is not identified by random assignment alone. The difficulty is that the potential outcome in the second term is evaluated at the counterfactual score $X_i(h)$, whereas the conditioning event is defined in terms of the potential score under the alternative threshold, $X_i(l)$.

\subsection{Point Identification under Rank-Invariance Assumption}\label{subsec:ident}
To achieve identification of $\E{Y_i(h,X_i(h),0)\mid X_i(l)=l}$ in \eqref{eq:second-l}, we have to impose some restriction that bridges the two threshold-policy regimes.
To this aim, we start with the rank-invariance assumption, which is commonly employed in various econometric problems, including quantile treatment-effect models and nonseparable triangular models (e.g., \citealp{Koenker01, Chernozhukov_Hansen:2005, Torgovitsky:2015}), and therefore provide a natural starting point:
\begin{assumption}\label{assumption:rank}
    Let $U_i\sim\text{Uniform}[0,1]$ be the latent rank.
    Then, $X_i(c)$ can be written as $X_i(c) = Q_c(U_i), c\in\{l,h\}$, where $Q_c$ is a strictly increasing quantile function of $X_i(c)$.
\end{assumption}
\begin{remark}
    When the score variable is invariant to the threshold (i.e., $X_i(c) = X_i$ for all $c$), rank invariance holds trivially. For instance, in educational settings where eligibility is determined by predetermined or inelastic characteristics---such as parental income that is difficult to adjust in the short run---individual scores do not respond to changes in the cutoff, rendering rank invariance a natural condition.
\end{remark} 

Assumption \ref{assumption:rank} requires the existence of a policy-invariant latent rank $U_i$ such that $X_i(l)$ and $X_i(h)$ are monotone transformations of the same scalar heterogeneity. Thus, changing the threshold may shift or stretch the distribution of scores, but it does not reshuffle individuals’ ranks in that distribution. 
In an educational setting, the basic intuition is not implausible: students who perform relatively well under one threshold regime also tend to perform relatively well under the other regime, and the same is true for students who perform relatively poorly. 

Rank invariance may provide a useful approximation when realized scores consist primarily of persistent individual heterogeneity together with relatively small idiosyncratic shocks. For example, if test scores are largely determined by stable differences in academic ability and effort, while idiosyncratic test-day shocks have a thin-tailed distribution and are small relative to the cross-sectional dispersion in these persistent components, rank reversals across threshold regimes may be limited. In such settings, individuals' relative positions in the score distribution may be largely preserved.
Conditional on covariates, rank invariance may be more plausible.\footnote{All subsequent identification arguments extend directly to this setting by interpreting the relevant distributions, quantiles, and causal effects as conditional on those covariates. The resulting conditional effects can then be averaged over the covariate distribution at the relevant cutoff to recover the unconditional effects.}


Let $F_c$ denote the distribution function of $X_i(c)$. 
Define
\begin{align*}
    R_i^c = F_c(X_i(c)) \quad\text{and}\quad u_c = F_c(c).
\end{align*}
Then, the second term of the right-hand side of \eqref{eq:second-l} can be rewritten as
\begin{align*}
    \E{Y_i(h,X_i(h),0)\mid X_i(l)=l} 
    = \E{Y_i(h,X_i(h),0)\mid R_i^l=u_l}.
\end{align*}
where the equality uses Assumption \ref{assumption:cont}(ii).
Here, the rank-invariance assumption implies $R_i^l = U_i = R_i^h$ almost surely, which in turn provides
\begin{align}
    \E{Y_i(h,X_i(h),0)\mid R_i^l=u_l} 
    = \E{Y_i(h,X_i(h),0) \mid R_i^h = u_l}.\label{eq:key}
\end{align}
Finally, we obtain that
\begin{align*}
    \E{Y_i(h,X_i(h),0) \mid R_i^h = u_l} &= \E{Y_i(h,X_i(h),0) \mid X_i(h) = Q_h(u_l), C_i = h}\\
    &= \E{Y_i \mid X_i = Q_h(u_l), C_i = h},
\end{align*}
where the last equality uses $F_c(c) = \P{X_i(c) \leq c} = \P{D_i(c) = 0}$ so that $\P{D_i(l) = 1} > \P{D_i(h) = 1} \Leftrightarrow F_l(l)<F_h(h) \Leftrightarrow u_l<u_h$. 
This implies that individuals at rank $u_l$ under the higher-threshold policy regime do not receive treatment.

Under Assumption \ref{assumption:rct}, $F_c$ is identified for both $c\in\{l,h\}$, and so is $Q_c$. Hence, the last term of the previous math display is identified, which in turn implies the identification of \eqref{eq:second-l}. An analogous argument applies $\theta_{h,1}(h)$, so that we obtain the following identification results. The proofs are deferred to Appendix \ref{sec:proof}.
\begin{prop}\label{prop:ident}
    Suppose the setup introduced in Section \ref{subsec:setup}.
    Suppose also that Assumptions \ref{assumption:rct}--\ref{assumption:cont} hold true.
    Let $F_c^{\mathtt{obs}}(x) = \P{X_i \leq x \mid C_i = c}$ denote the distribution function of observed score.
    Then the average treatment-assignment effect at the threshold is identified as, for $c\in\{l,h\}$,
    \begin{align*}
        \tau_c(c)
        =
            \lim\limits_{v\downarrow F_c^{\mathtt{obs}}(c)}\E{Y_i\mid F_c^{\mathtt{obs}}(X_i)=v, C_i=c} - \lim\limits_{v\uparrow F_c^{\mathtt{obs}}(c)}\E{Y_i \mid F_c^{\mathtt{obs}}(X_i)=v, C_i=c}
    \end{align*}
    provided that $0<F_c(c)<1$. If $F_c(c)\in\{0,1\}$, then $\tau_c(c)$ is not identified in general.

    In addition, suppose that Assumptions \ref{assumption:propensity}--\ref{assumption:rank} hold.
    Then the average incentive-response effect at the threshold is identified as
    \begin{align*}
        \theta_{l,0}(l)&=\lim\limits_{v\uparrow F_l^{\mathtt{obs}}(l)}\E{Y_i\mid F_l^{\mathtt{obs}}(X_i)=v, C_i=l} - \E{Y_i \mid F_h^{\mathtt{obs}}(X_i)=F_l^{\mathtt{obs}}(l), C_i=h}
    \end{align*}
    provided that $0<F_l(l)$.
    Similarly,
    \begin{align*}
        \theta_{h,1}(h)&=\E{Y_i\mid F_l^{\mathtt{obs}}(X_i)=F_h^{\mathtt{obs}}(h), C_i=l} - \lim\limits_{v\downarrow F_h^{\mathtt{obs}}(h)}\E{Y_i \mid F_h^{\mathtt{obs}}(X_i)=v, C_i=h}
    \end{align*}
    provided that $F_h(h)<1$.
    If $F_l(l)=0$, $\theta_{l,0}(l)$ is not identified in general. 
    If $F_h(h)=1$, $\theta_{h,1}(h)$ is not identified in general. 
\end{prop}
\begin{remark}
    As is clear from the preceding discussion, the key condition for identification is \eqref{eq:key}. In this sense, rank invariance can be viewed as a transparent sufficient condition for identification, and the result therefore applies more broadly than under rank invariance alone. 
\end{remark}
\begin{remark}\label{rem:rank-other}
    Related to the previous remark, one benefit of the stronger rank-invariance condition, rather than \eqref{eq:key}, is its greater identifying power. In particular, rank invariance implies that analogues of the equalities in \eqref{eq:key} hold at every rank level, rather than only at the cutoff ranks $u_l$. Consequently, the incentive-response effect is identified over the maximal rank regions permitted by the treatment-assignment rule: $\tilde \theta_0(u) := \E{Y_i(l,X_i(l),0)-Y_i(h,X_i(h),0)\mid U_i=u}$ is identified for $u\leq u_l$. Analogously, $\tilde \theta_1(u) := \E{Y_i(l,X_i(l),1)-Y_i(h,X_i(h),1)\mid U_i=u}$ is identified for $u\geq u_h$.
\end{remark}

The average treatment-assignment effect at the threshold $\tau_c(c)$ is precisely the same as the conventional RD treatment effect. 
This follows directly from
\begin{align*}
    \tau_c(c)
    &=
    \lim\limits_{v\downarrow F_c^{\mathtt{obs}}(c)}\E{Y_i\mid F_c^{\mathtt{obs}}(X_i)=v, C_i=c} - \lim\limits_{v\uparrow F_c^{\mathtt{obs}}(c)}\E{Y_i \mid F_c^{\mathtt{obs}}(X_i)=v, C_i=c}\\
    &=
    \lim\limits_{x\downarrow c}\E{Y_i\mid X_i=x, C_i=c} - \lim\limits_{x\uparrow c}\E{Y_i \mid X_i=x, C_i=c}
\end{align*}

The two effects, average treatment-assignment effect and average incentive-response effect, provide a more comprehensive picture of the threshold policy. 
More specifically, comparing these effects sheds light on the relative importance of the two channels.
For example, a positive RD treatment effect $\tau_l(l)$ may be accompanied by a negative and sufficiently large (in absolute value) incentive-response effect $\theta_{l,0}(l)$. 
In that case, the policy with threshold $l$ may reduce outcomes relative to the policy with threshold $h$, despite a positive RD treatment effect.

More precisely, we can decompose the ``total treatment effect'' into these two channels.
We define the individual total treatment effect by $Y_i(l,X_i(l),D_i(l))-Y_i(h,X_i(h),D_i(h))$.
Then the local average total treatment effect is given by
\begin{align*}
    \Delta_c(x) = \E{Y_i(l,X_i(l),D_i(l))-Y_i(h,X_i(h),D_i(h)) \mid X_i(c) = x},
\end{align*}
which is identified for all $x$ under rank invariance.
At the threshold $l$, we can decompose $\Delta_l(l)$ as
\begin{align}
    \Delta_l(l)
    &=\E{Y_i(l,X_i(l),D_i(l))-Y_i(l,X_i(l),D_i(h))\mid X_i(l) = l}\notag\\
    &\quad\qquad+ \E{Y_i(l,X_i(l),D_i(h))-Y_i(h,X_i(h),D_i(h))\mid X_i(l) = l}\notag\\
    &=\E{Y_i(l,l,1)-Y_i(l,l,0)\mid X_i(l) = l}
    + \E{Y_i(l,l,0)-Y_i(h,X_i(h),0)\mid X_i(l) = l}\notag\\
    &= \underbrace{\tau_l(l)}_{\text{RD treatment-assignment effect at }X_i(l)=l} 
    + \underbrace{\theta_{l,0}(l)}_{\text{effect operating through the incentive change from $h$ to $l$}},\label{eq:decomp}
\end{align}
where the second equality uses Assumptions \ref{assumption:propensity} and \ref{assumption:rank}. Similarly, at the higher threshold $h$, $\lim_{x\uparrow h}\Delta_h(x)$ is decomposed as
    \begin{align*}
        \lim_{x\uparrow h}\Delta_h(x)
        = \E{Y_i(l,X_i(l),1)-Y_i(h,h,0) \mid X_i(h) = h}
        = \tau_h(h) + \theta_{h,1}(h),
    \end{align*}
These decompositions permit a direct comparison of the two channels, yielding a richer analysis of the consequences of threshold policies and more informative policy implications.

\subsection{Relaxing Rank-Invariance Assumption}\label{subsec:partial}

Although rank invariance provides a natural benchmark, formally justifying its empirical validity may be challenging. Empirical researchers may therefore benefit from a sensitivity-analysis framework or a bounding approach. This subsection develops such a framework.
To simplify the discussion and notation, we focus on identification at the lower threshold, namely, $\theta_{l,0}(l)$, since an analogous argument applies to $\theta_{h,1}(h)$.

To proceed, we introduce the following mild continuity condition:
\begin{assumption}\label{assumption:cont-rank}\quad
    \begin{description}
        \item[(i)] $\E{Y_i(h,X_i(h), 0) \mid R_i^l = u, R_i^h = r}$ is continuous, relative to $\mathrm{supp}(R_i^l,R_i^h)$, at every point $(u_l, r)\in \mathrm{supp}(R_i^l,R_i^h)$.

        \item[(ii)] $u\mapsto\P{R_i^h \in A \mid R_i^l = u}$ is weakly continuous at $u=u_l$, that is,
        \begin{align*}
            \lim_{u\to u_l}\E{\phi(R_i^h) \mid R_i^l = u} = \E{\phi(R_i^h) \mid R_i^l = u_l}
        \end{align*}
        for every bounded continuous map $\phi:[0,1]\to\mathbb R$.
    \end{description}
\end{assumption}
Part (i) is a standard continuity condition. Part (ii) merits some explanation.
Because $R_i^l$ is continuously distributed, the event $R_i^l=u_l$ has probability zero, and hence the conditional distribution of $R_i^h$ given $R_i^l=u_l$ is not uniquely determined by the joint distribution of $(R_i^l,R_i^h)$. 
Assumption~\ref{assumption:cont-rank} rules out irregular choices of versions at this probability-zero event by requiring the conditional distribution at $u_l$ to coincide with the weak limit of the conditional distributions in its neighborhood. 
Without this restriction, the same joint distribution of $(R_i^l,R_i^h)$ could be associated with different versions of the conditional distribution at $u_l$, leading to different values of conditional expectations of functions of $R_i^h$ given $R_i^l=u_l$. 
The assumption therefore ensures that the conditional objects used below are well defined as local limits rather than depending on an arbitrary choice of version.

\subsubsection{Mean-Rank Invariance and Rank Sufficiency}
We now introduce the following condition as an alternative to rank invariance:

\begin{assumption}\label{assumption:mean-rank}
    \quad
    \begin{description}
        \item[(i)] $\E{R_i^h \mid R_i^l = u_l} = u_l$.

        \item[(ii)] $\E{Y_i(h,X_i(h),0) \mid R_i^l = u_l, R_i^h} = \E{Y_i(h,X_i(h),0) \mid R_i^h}$.
    \end{description}
\end{assumption}

Assumption \ref{assumption:mean-rank} is weaker than the rank-invariance assumption (Assumption \ref{assumption:rank}).
Under rank invariance, $R_i^h=R_i^l$ almost surely, so both parts of Assumption \ref{assumption:mean-rank} hold immediately. Assumption \ref{assumption:mean-rank}, however, allows individuals' ranks to change across the two policy regimes.

Part (i) imposes a (local) \textit{mean-rank invariance} condition. 
Individuals whose rank under the lower-threshold regime is $u_l$ may move to higher or lower ranks under the higher-threshold regime, but the condition requires these counterfactual ranks to be ``centered'' around $u_l$ on average.

Part (ii) is a \textit{rank sufficiency} restriction. Once an individual's rank under the higher-threshold regime is given, the rank under the lower-threshold regime contains no additional information about the mean potential outcome that would arise under the higher-threshold regime without treatment. 

Loosely speaking, these conditions are related in spirit to the rank-similarity assumption \citep{Chernozhukov_Hansen:2005}. Rank similarity allows individual ranks to vary across regimes while limiting systematic differences in their distribution across regimes. Our conditions impose a related, though distinct, restriction on cross-regime rank variation: mean-rank invariance rules out a systematic shift in the mean counterfactual rank at the cutoff, while rank sufficiency rules out systematic association between the remaining rank variation and the counterfactual outcome.

In the context of our empirical application discussed in Section~\ref{sec:motiv}, part (i) requires that among students located at the scholarship-cutoff rank $u_l$ under the lower-threshold regime, their expected rank under the higher-threshold regime is again $u_l$. 
While exact rank preservation may be violated in practice due to idiosyncratic exam shocks and heterogeneous score adjustments to the cutoff change, individual ranks are typically not reshuffled at random; rather, they exhibit substantial, albeit imperfect, persistence across settings (e.g., \citealp{Breit2024}). 
In such environments, mean-rank invariance may be plausible.
Part (ii) posits that, conditional on students' rank under the higher-threshold regime, their rank under the lower-threshold regime provides no additional information about their mean subsequent outcome under the higher-threshold regime in the absence of scholarship receipt. This may be reasonable when the no-treatment (i.e., higher-threshold) rank captures the persistent characteristics, such as academic preparedness and ability, that are relevant for subsequent academic outcomes, while the discrepancy between the two ranks primarily reflects idiosyncratic exam-day noise and regime-specific score adjustments that have little bearing on subsequent academic performance.

\begin{remark}
    There are many other potentially useful relaxations of the strict rank-invariance assumption. Related restrictions have been proposed in the quantile-treatment-effects literature and, more broadly, in the literature on the distribution of treatment effects. For example, \cite{Frandsen_Lefgren:2021} introduce a stochastic increasingness assumption as an alternative relaxation of rank invariance in the context of bounding the  treatment-effects distribution. Investigating the identifying implications and empirical relevance of such restrictions, or of more general restrictions on the dependence between counterfactual ranks, in our setting would be an interesting direction for future research.
\end{remark}

\subsubsection{Partial Identification}
We now establish identification under Assumption \ref{assumption:mean-rank}.
Recall that our main task is to identify the counterfactual mean function $\mathbb{E}[Y_i(h,X_i(h),0)\mid R_i^l=u_l]$ using the conditional mean function under the higher regime,
\begin{align*}
    m_h(r) := \mathbb{E}[Y_i(h,X_i(h),0) \mid R_i^h = r].
\end{align*}
An important issue, which does not arise under Assumptions~\ref{assumption:propensity}--\ref{assumption:rank}, is that $m_h(r)$ may not be identified over the entire range of ranks relevant for recovering $\mathbb{E}[Y_i(h,X_i(h),0)\mid R_i^l=u_l]$.
In particular, we need $m_h(r)$ to be identified over the support of $R_i^h$ conditional on $R_i^l=u_l$.
We therefore restrict our attention to the following case:

\begin{assumption}\label{assumption:high}
    Let $\mathcal S = \mathrm{supp} (R_i^h \mid R_i^l = u_l)$.
    Then $m_h(r)$ is identified over a known closed interval $\bar{\mathcal S}\subseteq[0,1]$ such that $\mathcal S \subseteq \bar{\mathcal S}$.
\end{assumption}
\begin{remark}\label{rem:support}
    In applications with $h=\infty$, Assumption \ref{assumption:high} holds trivially with $\bar{\mathcal S}=[0,1]$. 
\end{remark}

Without some extrapolation assumptions (e.g., parametric functional form assumption), this condition requires $\mathcal S \subseteq \bar{\mathcal S} \subseteq[0,u_h)$. That is, no individual at rank $u_l$ under the lower-threshold regime receives treatment under the higher-threshold regime. In particular,
\begin{align}
    \P{D_i(h)=0\mid R_i^l=u_l}=1.\label{eq:no-treat}
\end{align}
This support condition may be plausible when the two cutoff ranks, $u_l$ and $u_h$, are sufficiently far apart. In that case, moderate cross-regime rank movements can be accommodated without causing individuals located at the lower cutoff under regime $l$ to receive treatment under regime $h$.
We emphasize that both \eqref{eq:no-treat} and Assumption~\ref{assumption:high} are automatically satisfied in the important special case where $h=\infty$, as noted in Remark~\ref{rem:support}.

We are now in a position to establish identification.
Under Assumption \ref{assumption:mean-rank}(ii), 
\begin{align*}
    \E{Y_i(h,X_i(h),0) \mid R_i^l = u_l}=\E{m_h(R_i^h) \mid R_i^l = u_l}.
\end{align*}
Take an arbitraty convex function $g$ such that $g(r) \leq m_h(r)$. Then it follows that
\begin{align*}
    \E{m_h(R_i^h) \mid R_i^l = u_l} \geq \E{g(R_i^h) \mid R_i^l = u_l}
    \geq g\left(\E{R_i^h \mid R_i^l = u_l}\right) 
    = g(u_l),
\end{align*}
where the first inequality follows by definition of $g$, the second inequality uses Jensen's inequality, and the last equality uses Assumption \ref{assumption:mean-rank}(i).
Hence, the lower bound is given by
\begin{align*}
    \E{Y_i(h,X_i(h), 0) \mid R_i^l = u_l} 
    \geq 
    \sup_{g \in \mathcal F_{\mathtt{conv}}^{-}} g(u_l),
\end{align*}
where $\mathcal F_{\mathtt{conv}}^{-} = \left\{f: f \text{ is a convex function on } \bar{\mathcal S} \text{ such that } f \leq m_h\right\}$.
That is, the lower bound is given by the value of the convex envelope of $m_h$ at $u_l$.
A similar reasoning yields the upper bound as well:
\begin{align*}
    \E{Y_i(h,X_i(h), 0) \mid R_i^l = u_l} 
    \leq 
    \inf_{g \in \mathcal F_{\mathtt{conc}}^{+}} g(u_l),
\end{align*}
where $\mathcal F_{\mathtt{conc}}^{+} = \left\{f: f \text{ is a concave function on } \bar{\mathcal S} \text{ such that } f \geq m_h\right\}$.

Therefore, we can obtain the valid lower and upper bounds on $\mathbb{E}[Y_i(h,X_i(h),0)\mid R_i^l=u_l]$. The following proposition shows that the bounds are sharp as well.
\begin{prop}\label{prop:bounds}
    Under Assumptions \ref{assumption:rct}--\ref{assumption:propensity}, \ref{assumption:cont-rank}--\ref{assumption:high}, the identified set for $\theta_{l,0}(l)$ is given by
    \begin{align*}
        \bigg[
        &\lim\limits_{v\uparrow F_l^{\mathtt{obs}}(l)}\E{Y_i\mid F_l^{\mathtt{obs}}(X_i)=v, C_i=l} - \inf_{g \in \mathcal F_{\mathtt{conc}}^{+}} g(u_l),\\
        &\qquad \lim\limits_{v\uparrow F_l^{\mathtt{obs}}(l)}\E{Y_i\mid F_l^{\mathtt{obs}}(X_i)=v, C_i=l} - \sup_{g \in \mathcal F_{\mathtt{conv}}^{-}} g(u_l)
        \bigg],
    \end{align*}
    provided that $F_l(l) > 0$.
    Furthermore, this is the sharp identified set for $\theta_{l,0}(l)$.
\end{prop}
\begin{remark}
    A symmetric argument yields an analogous identified set for $\theta_{h,1}(h)$ after appropriately modifying Assumptions~\ref{assumption:mean-rank}--\ref{assumption:high}. We note, however, that when identifying $\theta_{h,1}(h)$, setting $h=\infty$ does not guarantee the counterpart of Assumption~\ref{assumption:high}, namely, identification of $m_l(r):=\mathbb{E}[Y_i(l,X_i(l),1)\mid R_i^l=r]$, because this function concerns the mean potential outcome under treatment. To obtain an analogous simple sufficient condition, an experimenter may instead include a full-treatment arm, corresponding to $l=-\infty$.
\end{remark}

Proposition~\ref{prop:bounds} characterizes the sharp identified set for $\theta_{l,0}(l)$ under the mean-rank-invariance and rank-sufficiency assumptions, together with Assumption~\ref{assumption:high}, which ensures identification of $m_h$. Because the treatment-assignment effect $\tau_l(l)$ remains point identified, an analogous decomposition can be carried out provided that \eqref{eq:no-treat} holds. In particular, the second equality in \eqref{eq:decomp} continues to hold under \eqref{eq:no-treat}. Thus, the same interpretation of the decomposition and comparison between $\tau_l(l)$ and $\theta_{l,0}(l)$ remain valid, although $\theta_{l,0}(l)$ is now only partially identified.

The identified set is easy to compute.
Define, for $a\le u_l \le b$,
\begin{align*}
    \Xi(a,b \mid m_h, u_l) := 
    \begin{cases}
        \frac{b-u_l}{b-a} m_h(a) + \left(1 -\frac{b-u_l}{b-a}\right)m_h(b) & (a<b),\\
        m_h(u_l) & (a=b).
    \end{cases}.
\end{align*}
For $a<b$, $\Xi(a,b\mid m_h,u_l)$ is the value at $u_l$ of the affine function whose graph is the chord connecting the two points $(a,m_h(a))$ and $(b,m_h(b))$ on the graph of $m_h$. The case $a=b=u_l$ corresponds to the degenerate chord consisting of the single point $(u_l,m_h(u_l))$.
Then, by a standard result from convex analysis (\citealp[Corollary 17.1.5]{Rockafellar70}), we have that
\begin{align*}
    \sup_{g \in \mathcal F_{\mathtt{conv}}^{-}} g(u_l) = \inf_{\substack{a,b\in\bar{\mathcal S}\\a\le u_l\le b}} \Xi(a,b \mid m_h, u_l)\quad\text{ and }\quad
    \inf_{g \in \mathcal F_{\mathtt{conc}}^{+}} g(u_l) = \sup_{\substack{a,b\in\bar{\mathcal S}\\a\le u_l\le b}} \Xi(a,b \mid m_h, u_l).
\end{align*}
That is, the value of the convex envelope of $m_h$ at $u_l$ can be computed as the infimum of the values at $u_l$ of all chords of $m_h$ whose endpoints bracket $u_l$. Analogously, the value of the concave envelope of $m_h$ at $u_l$ is obtained by taking the supremum over the same set of chords.

\section{Empirical Application}\label{sec:emp}

We illustrate the importance of disentangling the two effects of a threshold policy---the treatment-assignment channel and the incentive-response channel---using experimental data from \cite{Berry22}. Section \ref{sec:motiv} describes the experimental setting and its institutional background. The main text focuses on students who were in sixth grade in 2015, for whom the distinction between the two effects is particularly salient. Results for fifth-grade students and a comparison between the fifth- and sixth-grade samples are reported in Appendix~\ref{sec:addemp} of the Online Appendix.

In this application, students assigned to the scholarship arm face a threshold $l$ corresponding to the top 15 percent of scores on the June 2015 examination. Students assigned to the control arm face a threshold $h=\infty$, since no student in that arm can receive the award. We therefore focus on the treatment-assignment effect at the lower threshold, $\tau_l(l)$, and on the incentive-response effect $\theta_{l,0}(l)$, which are the identified quantities in this setting.
These quantities are identified from conditional expectations of the March 2016 test score given rank. Figure~\ref{fig} plots the corresponding conditional expectation functions. The solid line shows $\widehat{\mathbb{E}}[Y_i\mid U_i=u,C_i=l]$ whereas the dashed line shows $\widehat{\mathbb{E}}[Y_i\mid U_i=u,C_i=h]$.

\begin{figure}[t]
    \begin{center}
    \includegraphics[width=0.75\linewidth]{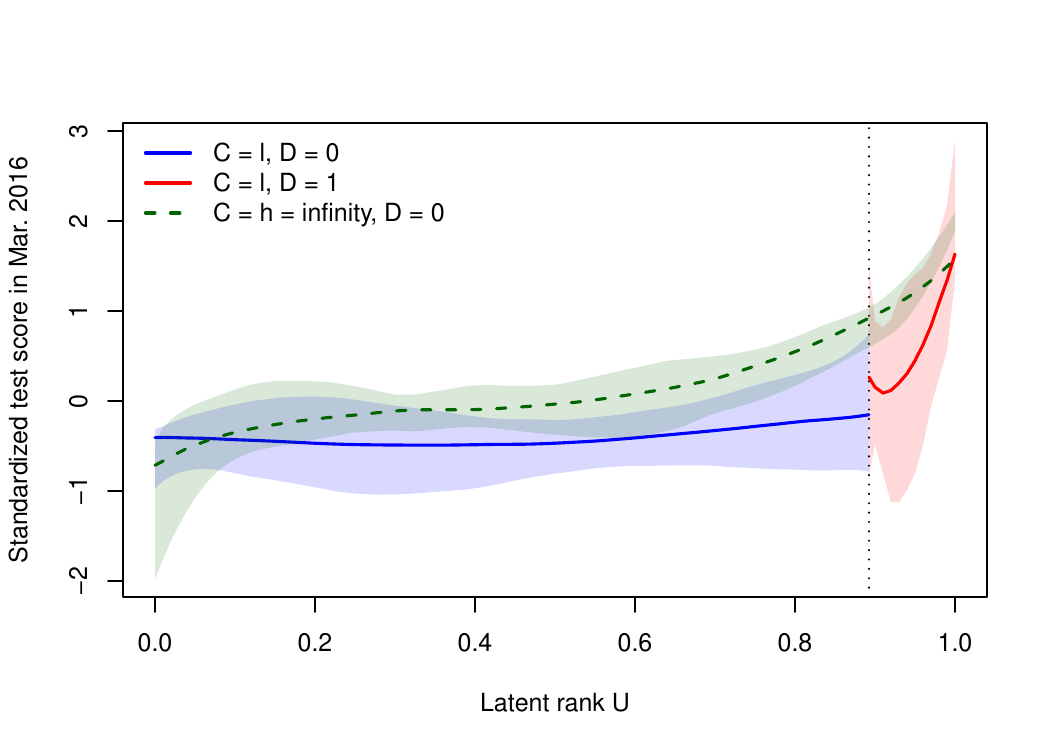}
    \caption{Expected Test Scores Conditional on the Rank}
    \label{fig}
    \end{center}

    \footnotesize
    \renewcommand{\baselineskip}{11pt}
    \textbf{Note:} The estimated IMSE-optimal bandwidths are used ($0.368$ for $C_i=l, D_i=0$; $0.058$ for $C_i=l, D_i=1$; $0.286$ for $C_i=h, D_i=0$). The shaded regions indicate the 95\% robust-bias corrected (RBC) confidence intervals.
    The robust-bias corrected standard error is clustered at school level. The cutoff on the latent-rank scale is approximately $0.89$. The top-15-percent criterion was computed among all students, including students in the control arm and those excluded from our analysis. The corresponding cutoff in our analysis sample therefore need not equal $0.85$.
\end{figure}

We first consider the treatment-assignment effect $\tau_l(l)$, which corresponds to the conventional RD estimand at the cutoff. The point estimates and $95\%$ robust bias-corrected confidence intervals are reported in the last two columns of the first row of Table~\ref{tab:te}. We consider two bandwidth choices: the IMSE-optimal bandwidth used in Figure~\ref{fig} and the MSE-optimal bandwidth for the RD estimand proposed by \cite{Imbens_Kalyanaraman:2011}.
Consistent with the discontinuity in Figure~\ref{fig}, the RD point estimate is positive. It is statistically significant when we use the RD-specific bandwidth. Taken at face value, the conventional RD estimand suggests that receiving the financial award may improve subsequent test performance for students near the cutoff.

However, the policy implications become more nuanced once we consider the incentive-response effect $\theta_{l,0}(l)$. 
Figure~\ref{fig} suggests that students in the scholarship arm exhibit lower subsequent test scores over much of the relevant rank distribution. 
Under rank invariance, this pattern implies a negative incentive-response effect at the cutoff: assignment to the threshold policy may lower subsequent achievement even holding scholarship receipt fixed.

\begin{table}[t]
    \begin{center}
    \scalebox{0.90}{
    \begin{tabular}{lccccc}\hline\hline
        \multirow{2}{*}{Rank $u$} & \multirow{2}{*}{$0.0$} & \multirow{2}{*}{$0.4$} & \multirow{2}{*}{$0.8$} & \multicolumn{2}{c}{$u_l \approx 0.89$}\\
         & & & & (IMSE BW) & (RD BW)\\\hline
        Treatment & -- & -- & -- & $0.418$ & $0.424$\\
        \,\,\, assignment & -- & -- & -- & [$-1.01, 1.73$] & [$0.60, 3.02$]\\
        Incentive & $0.308$ & $-0.391$ & $-0.782$ & $-1.073$ & $-1.109$\\
        \,\,\, response & [$-0.26, 1.41$] & [$-0.97, -0.05$] & [$-1.27, -0.08$] & [$-1.63, -0.05$] & [$-2.67, -1.32$]\\
         \multirow{2}{*}{Total} & $0.308$ & $-0.391$ & $-0.782$ & $-0.655$ & $-0.685$\\
         & [$-0.26, 1.41$] & [$-0.97, -0.05$] & [$-1.27, -0.08$] & [$-1.67, 0.71$] & [$-1.26, 0.89$]\\\hline
    \end{tabular}}
    \caption{Treatment Effects}
    \label{tab:te}
    \end{center}

    \footnotesize
    \renewcommand{\baselineskip}{11pt}
    \textbf{Note:} 
    The conventional point estimates and 95\% robust bias-corrected confidence intervals are reported. Because the confidence intervals incorporate bias correction, they are not necessarily centered on the corresponding point estimates.
    Except for the last column, the estimated IMSE bandwidths, reported in the note below Figure \ref{fig}, are used. In the last column (RD BW), the MSE-optimal bandwidth $0.082$ for RD estimator proposed in \cite{Imbens_Kalyanaraman:2011} is used.
    The robust-bias corrected standard error is clustered at school level. 
\end{table}

Table~\ref{tab:te} reports the estimation results under rank-invariance assumption.
The effects via two channels, $\tau_l(l)$ and $\theta_{l,0}(l)$, are presented in the final two columns.
At this rank, the estimated incentive-response effect is more than twice as large in absolute value as the treatment-assignment effect and is statistically significantly negative. 
Accordingly, the total effect at this rank $\Delta_l(l)$ is estimated to be negative.
Thus, the overall effect of the threshold-based scholarship policy may be negative, even though the conventional RD estimate is positive. 

Moreover, we find that the estimated incentive-response effects at other rank levels---which are identified under rank invariance (Remark~\ref{rem:rank-other})---are negative and statistically significant at $u=0.4$ and $u=0.8$. This finding suggests that the adverse effect observed at the cutoff may not be unique to the threshold but may extend over a broader range of ranks.

\begin{figure}[t]
    \begin{center}
    \includegraphics[width=0.75\linewidth]{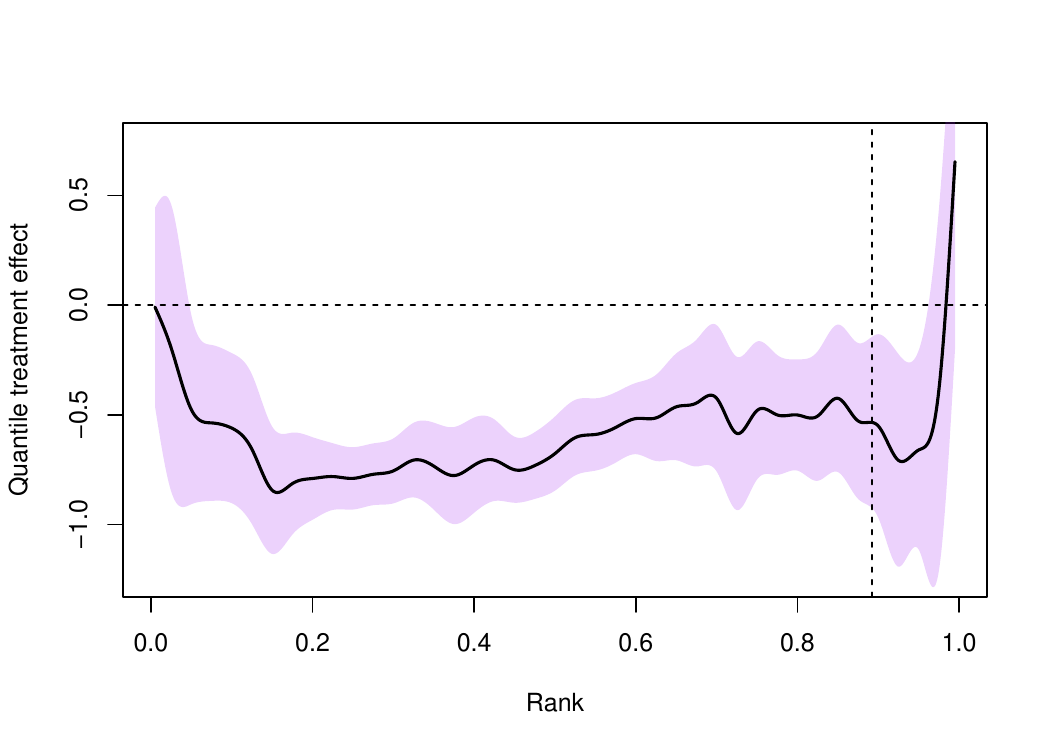}
    \caption{Quantile Treatment Effect of Threshold Assignment on Test Score in June 2015}
    \label{qte}
    \end{center}

    \footnotesize
    \renewcommand{\baselineskip}{11pt}
    \textbf{Note:} The shaded region indicates the 95\% bootstrap confidence band. 
\end{figure}

These patterns suggest that the merit-based scholarship program induces a behavioral response which negatively affect educational attainment.
For example, \cite{Berry22} find that students' motivation to study at the time of the June 2015 exam---measured by how strongly they agree with the statement ``I am motivated to study hard'' on a five-point scale---was decreased except for the initially high-performing students.
This may perhaps be due to the stringent top-15-percent cutoff.\footnote{Threshold-based incentive schemes may encourage effort when the performance target is attainable, but may have weak or adverse effects when the target is difficult to reach (\citealp{Leuven10}). More broadly, the effects of financial incentives can depend on individuals' capabilities, task difficulty, and intrinsic motivation; see \cite{CamererHogarth99}.}
Consistent with this observation, the quantile treatment effect of threshold assignment on the June 2015 test score $X_i$, i.e., $Q_l(u)-Q_h(u)$ is negative over nearly the entire distribution and becomes positive only at extremely high quantiles (Figure~\ref{qte}).

Taken together with the positive estimate of $\tau_l(l)$ and the larger negative estimate of $\theta_{l,0}(l)$, these findings suggest that the financial-aid program may have discouraged students over much of the score distribution, thereby reducing their June 2015 test scores and, potentially, their human-capital accumulation. This adverse incentive response may in turn have lowered subsequent academic achievement, including test scores in 2016, even though scholarship receipt itself appears to have had a modest positive effect on academic achievement.
If so, a merit-based implementation may reduce subsequent educational achievement relative to an otherwise comparable policy without a merit threshold. 

\begin{figure}[t]
    \begin{center}
    \includegraphics[width=0.75\linewidth]{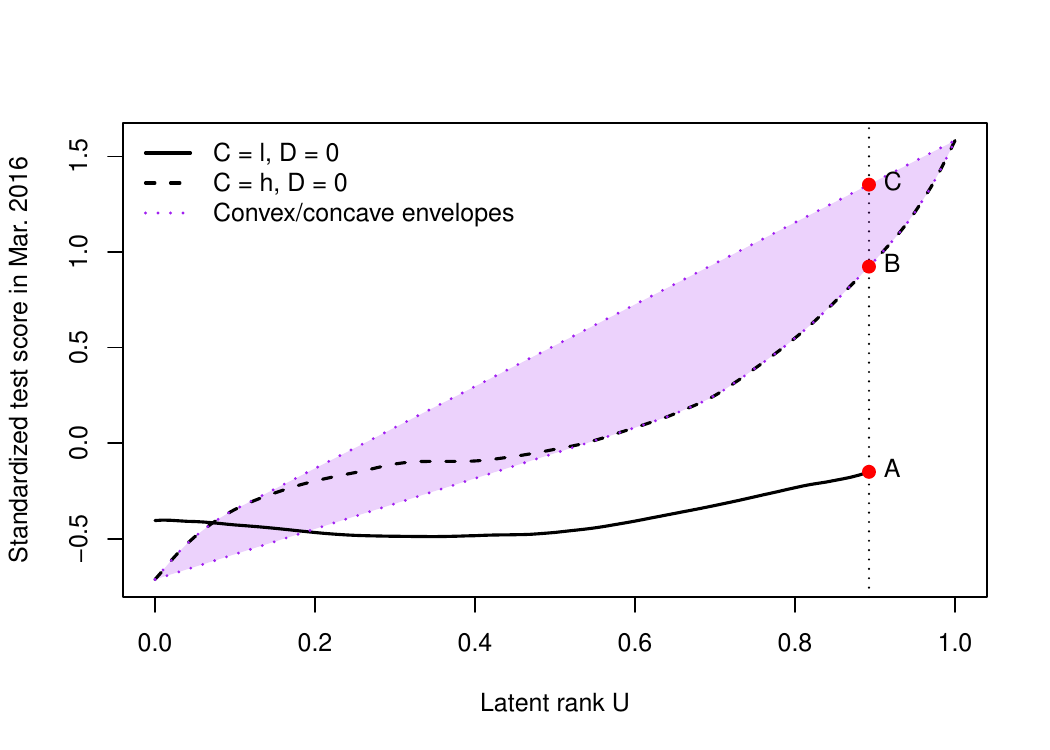}
    \caption{Envelopes and Bounds}
    \label{bounds}
    \end{center}

    \footnotesize
    \renewcommand{\baselineskip}{11pt}
    \textbf{Note:} The estimated IMSE-optimal bandwidths are used.
    The identified set is given by $[A-C, A-B]$.
\end{figure}

Finally, we compute the identified set under a weaker set of assumptions, namely, mean-rank invariance and rank sufficiency. Figure~\ref{bounds} reports the estimated convex and concave envelopes of $\mathbb{E}[Y_i\mid R_i^h=r,C_i=h]$, which yield the bounds $[A-C,A-B]=[-1.502,-1.073]$.
As shown in the figure, the lower convex envelope coincides with the observed regression function. Consequently, the upper bound coincides with the negative point estimate obtained under rank invariance, as reported in Table~\ref{tab:te}.

The bounds therefore allow for a more pessimistic conclusion. 
In particular, the true magnitude of this negative incentive-response effect may be larger in absolute value than the point estimate obtained under rank invariance.
Thus, these findings provide even stronger evidence than the point-identified results presented above that merit-based assignment may have important limitations as an incentive mechanism.

\section{Conclusion}\label{sec:concl}

This paper developed a framework for separately identifying the effects of threshold policies operating through two distinct channels: changes in treatment assignment and incentive-induced changes in individuals' response.
Under the rank-invariance-type conditions and mild regularity conditions, we showed that these effects can be locally identified.
Our empirical application illustrates the importance of accounting for both channels, as relying solely on the conventional RD effect can lead to misleading conclusions about the overall consequences of a threshold policy.

More broadly, this paper contributes to the literature by formalizing a potentially important causal channel that has received limited attention in conventional RD analyses. We view our framework as a starting point for further research on threshold policies. Although our analysis focuses on settings with random assignment of thresholds, extending the framework to observational settings is an important direction for future work. Some of our identifying assumptions may carry over directly, whereas others will need to be replaced or strengthened. It would also be useful to integrate our analysis with recent developments in the RD literature, such as multidimensional assignment rules.

\bibliographystyle{apalike} 
\bibliography{refs}

\newpage

\begin{center}
    \LARGE Online Appendices
\end{center}
\begin{center}
    \large Shunsuke Imai, Koshi Nishida, \& Yuta Okamoto
\end{center}
\appendix
\section{Proofs}\label{sec:proof}

\begin{proof}[Proof of Proposition \ref{prop:ident}]
    Recall $u_c = F_c(c)$.
    By Assumption \ref{assumption:rct}, for any $x\in \mathcal X$,
    \begin{align*}
        F_c(x) = \P{X_i(c) \le x} = \P{X_i \le x \mid C_i=c} = F_c^{\mathtt{obs}}(x),
    \end{align*}
    so that $F_c(x)$ is identified as $F_c^{\mathtt{obs}}(x)$.
    By Assumption \ref{assumption:cont}(ii), $F_c$ is continuous and strictly increasing on $\mathrm{int}(\mathcal X)$. Hence, $Q_c$ is the inverse of $F_c$. 
    It therefore follows that $R_i^c = F_c(X_i(c)) = F_c(Q_c(R_i^c))$.
    Then, $\mathbb{E}[Y_i(c,d,Q_c(u)) \mid R_i^c = u] = \mathbb{E}[Y_i(c,d,Q_c(u)) \mid X_i(c) = Q_c(u)]$ is continuous in $u\in(0,1)$ by Assumption \ref{assumption:cont}. This suggests the continuity of $\mathbb{E}[Y_i \mid F_{c}^{\mathtt{obs}}(X_i) = u, C_i=c]$ for $u<u_c$ and $u\geq u_c$.

    We begin by proving identification of $\tau_c(c)$. Note that
    \begin{align*}
        \tau_c(c) &= \E{Y_i(c,X_i(c),1)-Y_i(c,X_i(c),0)\mid X_i(c) = c}\\
        &= \E{Y_i(C_i,X_i,1)-Y_i(C_i,X_i,0)\mid R_i^c=u_c, C_i=c}\\
        &= \E{Y_i(C_i,X_i,1)\mid F_c^{\mathtt{obs}}(X_i)=u_c, C_i=c} - \E{Y_i(C_i,X_i,0)\mid F_c^{\mathtt{obs}}(X_i)=u_c, C_i=c}
    \end{align*}
    where the second equality uses Assumption \ref{assumption:rct}.
    By the continuity established above,
    \begin{align*}
         &\E{Y_i(C_i,X_i,1)\mid F_c^{\mathtt{obs}}(X_i)=u_c, C_i=c} - \E{Y_i(C_i,X_i,0)\mid F_c^{\mathtt{obs}}(X_i)=u_c, C_i=c}\\
         &= 
         \lim_{v\downarrow F_c^{\mathtt{obs}}(c)}\E{Y_i\mid F_c^{\mathtt{obs}}(X_i)=v, C_i=c} - \lim_{v\uparrow F_c^{\mathtt{obs}}(c)}\E{Y_i\mid F_c^{\mathtt{obs}}(X_i)=v, C_i=c}.
    \end{align*}
    We then show identification of $\theta_{l,0}(l)$. Note that 
    \begin{align*}
        \theta_{l,0}(l) &= \E{Y_i(l,X_i(l),0)-Y_i(h,X_i(h),0) \mid X_i(l) = l}\\
        &= \E{Y_i(l,X_i(l),0) \mid X_i(l) = l} - \E{Y_i(h,X_i(h),0) \mid R_i^h = u_l}\\
        &= \E{Y_i(l,X_i,0)\mid R_i^l=u_l, C_i=l} - \E{Y_i\mid R_i^h=u_l, C_i=h}\\
        &= \lim_{v \uparrow F_l^{\mathtt{obs}}(l)}\E{Y_i\mid F_l^{\mathtt{obs}}(X_i)=v, C_i=l} - \E{Y_i\mid F_h^{\mathtt{obs}}(X_i)=F_l^{\mathtt{obs}}(l), C_i=h},
    \end{align*}
    where the second equality uses Assumption \ref{assumption:rank}, and the third equality uses Assumptions \ref{assumption:rct} and \ref{assumption:propensity}.
    Identification of $\theta_{h,1}(h)$ can be shown analogously.
\end{proof}

\begin{proof}[Proof of Proposition \ref{prop:bounds}]
    The validity is shown in the main text. It remains to prove the sharpness.
    Recall $m_h(r) := \E{Y_i(h,X_i(h),0) \mid R_i^h = r}$. 
    It suffices to show that every value in $[\sup_{g \in \mathcal F_{\mathtt{conv}}^{-}} g(u_l),\inf_{g \in \mathcal F_{\mathtt{conc}}^{+}} g(u_l)]$ can be attained by $\E{Y_i(h,X_i(h),0) \mid R_i^l = u_l}$.

    We first construct the conditional distribution of $R_i^h$ given $R_i^l=u_l$.
    Under Assumption \ref{assumption:cont}, $m_h$ admits a continuous extension to $[0,1]$.
    Together with the compactness of $\bar{\mathcal S}$ (Assumption \ref{assumption:high}), there exists $a \leq u_l \leq b$ such that $a,b \in \bar{\mathcal S}$ and
    \begin{align*}
        \inf_{g \in \mathcal F_{\mathtt{conc}}^{+}} g(u_l) 
        =
        \pi m_h(a) + (1 - \pi) m_h(b),
        \qquad 
        \pi = \frac{b - u_l}{b - a},
    \end{align*}
    when $a<b$. If $a=u_l=b$, the right-hand side is understood as $m_h(u_l)$.
    Here, we use the chord representation established in the main text.
    
    Define the two-point distribution
    \begin{align*}
        P_U = \begin{cases}
            \pi \delta_a + (1-\pi) \delta_b & (a<b),\\
            \delta_{u_l} & (a=b).
        \end{cases}
    \end{align*}
    Then, it follows that $\int r \, d P_U(r) = u_l$ and $\int m_h(r)\, d P_U(r) = \inf_{g \in \mathcal F_{\mathtt{conc}}^{+}} g(u_l)$.
    Similarly, there exists two-point distribution $P_L$ such that $\int r \, d P_L(r) = u_l$ and $\int m_h(r)\, d P_L(r) = \sup_{g \in \mathcal F_{\mathtt{conv}}^{-}} g(u_l)$.
    Hence, by defining $P_\lambda^* = \lambda P_L + (1-\lambda)P_U$ with $\lambda\in[0,1]$, $\int r \, d P_\lambda^*(r) = u_l$ and
    \begin{align*}
        \int m_h(r)\, d P_\lambda^*(r) = \lambda \sup_{g \in \mathcal F_{\mathtt{conv}}^{-}} g(u_l) + (1-\lambda) \inf_{g \in \mathcal F_{\mathtt{conc}}^{+}} g(u_l).
    \end{align*}
    Therefore, there exists a conditional distribution $P_\lambda^*$ that is consistent with Assumption \ref{assumption:mean-rank}(i), $\mathrm{supp}(R_i^h \mid R_i^l = u_l) \subset \bar{\mathcal S}$, and attains $\mathbb{E}[m_h(R_i^h) \mid R_i^l = u_l] = \lambda \sup_{g \in \mathcal F_{\mathtt{conv}}^{-}} g(u_l) + (1-\lambda) \inf_{g \in \mathcal F_{\mathtt{conc}}^{+}} g(u_l)$ for any $\lambda \in[0,1]$, although this construction does not yet embed $P_\lambda^*$ into a joint distribution of $(R_i^h, R_i^l)$ with uniform marginals. We next provide such an embedding.

    Let $\{r_1,\ldots,r_J\}$ ($J\le4$) denote the finite support of $P_\lambda^*$ with associated probabilities $p_1,\ldots,p_J$. Note that $r_j \in \bar{\mathcal S}$ for every $j$ by construction.
    For each $r_j$, we can choose a Lebesgue-measure-preserving map $T_j: [0,1]\to[0,1]$ such that $T_j(u_l) = r_j$, and it is continuous at $u_l$.
    Such a map can be constructed as follows: When $r_j \in (0,1)$, we can take $T_j(u) = \{- u + r_j + u_l\}$, where $\{x\} = x - \lfloor x \rfloor$. This map is measure-preserving since $r_j + u_l$ is a constant. Since $r_j \in (0,1)$, it is also continuous at $u_l$.
    When $r_j = 0$, we begin by defining, for $c\in(0,1)$,
    \begin{align*}
        T_j^c(u) := \begin{cases}
            \{u - u_l\}/c &  \text{if } \{u - u_l\}\in[0,c],\\
            (1 - \{u - u_l\}) / (1-c)  &  \text{if } \{u - u_l\}\in(c,1].
        \end{cases}
    \end{align*}
    It is straightforward to see that the measure-preserving property and continuity are satisfied.
    Thus we may simply fix $c\in(0,1)$ and set $T_j = T_j^c$.
    For $r_j=1$, we can find $T_j$ by a similar construction.
    
    Put $R_i^l \sim \mathrm{Unif}[0,1]$ and let $J_i$ be a discrete random variable that is independent of $R_i^l$ and satisfies $\P{J_i = j} = p_j$.
    Define $R_i^h = T_{J_i}(R_i^l)$. Then it holds that $R_i^h \sim \mathrm{Unif}[0,1]$.
    Moreover, it follows that
    \begin{align*}
        \P{R_i^h \in A \mid R_i^l = u_l} = \sum_{j=1}^J p_j \mathbf{1}\{T_j(u_l) \in A\}
        = \sum_{j=1}^J p_j \mathbf{1}\{r_j \in A\} = P_\lambda^*(A).
    \end{align*}
    Furthermore, by letting $G(u,A):=\P{R_i^h \in A \mid R_i^l = u}$, for every continuous bounded function $\phi$,
    \begin{align*}
        \int \phi(r) G(u,dr) = \sum_{j=1}^J p_j \phi(T_j(u)) \to 
        \sum_{j=1}^J p_j \phi(T_j(u_l)) = \sum_{j=1}^J p_j \phi(r_j)
        = \int \phi(r) dP_\lambda^*
    \end{align*}
    as $u\to u_l$, implying that Assumption \ref{assumption:cont-rank}(ii) is satisfied.
    Hence, $R_i^h \mid (R_i^l = u_l) \sim P_\lambda^*$.
    Thus, $P_\lambda^*$ can be embedded into a joint distribution of $(R_i^h, R_i^l)$ with uniform marginals.

    We now define the potential outcomes $Y_i(c,x,d)$.
    It is convenient to start with the potential outcomes evaluated at the realized potential scores.
    Let $P_{\mathtt{obs}}$ denote the distribution of the observed variables (satisfying the observed-data implication of Assumption \ref{assumption:propensity}).
    By Assumption \ref{assumption:high}, $m_h(r)=\mathbb{E}[Y_i(h,X_i(h),0) \mid R_i^h = r]$ is identified from $P_{\mathtt{obs}}$ over $\bar{\mathcal S}$. Let $m_h^{\mathtt{obs}}(r)$ denote the identified function. 
    Note that $\mu_l := \lim_{r\uparrow u_l}\mathbb{E}[Y_i(l,X_i(l),0) \mid R_i^l = r]$ is also identified from $P_{\mathtt{obs}}$.
    
    Take any full-data distribution $P_0$ compatible with $P_{\mathtt{obs}}$ and satisfying Assumption \ref{assumption:cont} and $\mathbb{E}_{P_0}[Y_i(h,X_i(h),0) \mid R_i^h = r] = m_h^{\mathtt{obs}}(r)$ over $\bar{\mathcal S}$.
    For $c\in\{l,h\}$ and $d\in\{0,1\}$, let $\mu_{c,d}^0(x) := \mathbb{E}_{P_0}[Y_i(c,x,d) \mid X_i(c) = x]$, which is continuous by Assumption~\ref{assumption:cont}, and define $m_{c,d}^0(r) := \mu_{c,d}^0(Q_c(r))$. We use the continuous extension of these function to the boundary whenever necessary. In particular, $m_{h,0}^0(r) = m_h^{\mathtt{obs}}(r)$ for $r \in \bar{\mathcal S}$

    Define
    \begin{align*}
        \tilde{P}_c(A \mid r) := P_0[(Y_i(c,X_i(c),0), Y_i(c,X_i(c),1)) \in A \mid R_i^c = r]
    \end{align*}
    for each $c\in\{l,h\}$. 
    Because $\tilde{P}_c$ are defined only almost everywhere in $r$, we choose their versions such that, for every $r\in[0,1]$,
    \begin{align*}
        \int y_d \tilde{P}_c(dy_0, dy_1 \mid r_j) = m_{c,d}^0(r),
        \quad
        c\in\{l,h\}\,\,\text{and}\,\, d\in\{0,1\}.
    \end{align*}
    This modification affects $\tilde{P}_c$ only on sets of Lebesgue measure zero and therefore does not alter the corresponding joint distributions under $P_0$. In particular,    
    \begin{align*}
        \int y_0 \tilde{P}_h(dy_0, dy_1 \mid r_j) = m_h^{\mathtt{obs}}(r)\,\,\text{ for }\,\,r\in\bar{\mathcal S},
        \quad \text{ and }\quad
        \int y_0 \tilde{P}_l(dy_0, dy_1 \mid u_l) = \mu_l.
    \end{align*}
    
    Now take the joint distribution of $(R_i^l,R_i^h)$ constructed above.
    Conditional on $(R_i^l,R_i^h)$, generate 
    \begin{align*}
        \big(Z_{i,l,0}, Z_{i,l,1}\big) \mid (R_i^l,R_i^h)
        \sim 
        \tilde{P}_l(\cdot,\cdot \mid R_i^l),
    \end{align*}
    and conditional on $(R_i^l, R_i^h, Z_{i,l,0}, Z_{i,l,1})$, generate
    \begin{align*}
        \big(Z_{i,h,0}, Z_{i,h,1}\big) \mid (R_i^l, R_i^h, Z_{i,l,0}, Z_{i,l,1})
        \sim 
        \tilde{P}_h(\cdot,\cdot \mid R_i^h).
    \end{align*}
    Define $X_i(c) := Q_c(R_i^c)$ for $c\in\{l,h\}$.
    For every $x\in\mathcal X$ and $d\in\{0,1\}$, define the (full) potential outcomes by
    \begin{align*}
        Y_i(c,x,d) := 
        Z_{i,c,d} + \mu_{c,d}^0(x) - \mu_{c,d}^0(X_i(c)).
    \end{align*}
    Then, it satisfies $Y_i(c,X_i(c),d) = Z_{i,c,d}$.
    Next, define $D_i(c) := \mathbf{1}\{X_i(c)\ge c\}$ for $c\in\{l,h\}$.
    Draw $C_i$ independently to the potential variables constructed above, with the same assignment probability as under $P_{\mathtt{obs}}$, and let $X_i := X_i(C_i)$, $D_i := D_i(C_i)$, and $Y_i := Y_i(C_i, X_i(C_i), D_i(C_i))$.
    Let $P_\lambda^\circ$ denote the resulting joint distribution.
    By construction, the joint distribution of $(Y_i(c,X_i(c),0), Y_i(c,X_i(c),1), R_i^c) = (Z_{i,c,0}, Z_{i,c,1}, R_i-c)$ under $P_\lambda^\circ$ is identical to that under $P_0$ for each $c\in\{l,h\}$. Hence $P_\lambda^\circ$ generates the same observed-data distribution $P_{\mathtt{obs}}$.

    We now verify the constructed variables satisfy the maintained assumptions.
    Since $C_i$ is drawn independently, Assumption~\ref{assumption:rct} is satisfied.
    Noting that 
    \begin{align*}
        \mathbb{E}_{P_\lambda^\circ}[Y_i(c,x,d) \mid X_i(c) = x]
        = \mathbb{E}_{P_\lambda^\circ}[Z_{i,c,d} \mid X_i(c) = x]
        = \mu_{c,d}^0(x),
    \end{align*}
    we find that Assumption~\ref{assumption:cont}(i) is met.
    Assumption~\ref{assumption:cont}(ii) is also preserved because $R_i^c\sim \mathrm{Unif}[0,1]$ so that $X_i(c) = Q_c(R_i^c)$ has the same marginal distribution under $P_0$.
    The marginal distributions of $X_i(l)$ and $X_i(h)$ are unchanged, so the treatment propensities are unchanged as well and Assumption \ref{assumption:propensity} is preserved.
    We next verify Assumption \ref{assumption:cont-rank}(i). For every $(u,r)$ in the support of $(R_i^l,R_i^h)$,
    \begin{align*}
        \mathbb{E}_{P_\lambda^\circ}[Y_i(h,X_i(h),0)\mid R_i^l=u,R_i^h=r]
        =
        \mathbb{E}_{P_\lambda^\circ}[Z_{i,h,0}\mid R_i^l=u,R_i^h=r]
        =
        m_{h,0}^0(r).
    \end{align*}
    Since $m_{h,0}^0$ is continuous, this conditional mean is continuous, relative to $\mathrm{supp}(R_i^l,R_i^h)$, at every point $(u_l,r)$ in that support. Thus Assumption \ref{assumption:cont-rank}(i) is satisfied.
    Assumption \ref{assumption:cont-rank}(ii) is already verified.
    The construction also implies Assumption \ref{assumption:mean-rank}(ii):
    \begin{align*}
        \mathbb{E}_{P_\lambda^\circ}\left[Y_i(h,X_i(h),0) \mid R_i^l = u_l, R_i^h\right]
        &= m_{h,0}^0(R_i^h) 
        = \mathbb{E}_{P_\lambda^\circ}[Y_i(h,X_i(h),0) \mid R_i^h].
    \end{align*}
    Furthermore, Assumption \ref{assumption:mean-rank}(i) and the support condition in Assumption \ref{assumption:high} are inherited from the properties of $P_\lambda^*$.
    
    Finally, by construction, $\mathbb{E}_{P_\lambda^\circ}[Y_i(l,X_i(l),0)\mid R_i^l=u_l] = \mu_l$ and hence $\theta_{l,0}(l) = \mu_l - \mathbb{E}_{P_\lambda^\circ}[Y_i(h,X_i(h),0)\mid R_i^l=u_l]$.
    Note that
    \begin{align*}
        &\mathbb{E}_{P_\lambda^\circ}\left[Y_i(h,X_i(h),0) \mid R_i^l = u_l\right]
        = 
        \mathbb{E}_{P_\lambda^\circ}\left[
        \mathbb{E}_{P_\lambda^\circ}\left[Y_i(h,X_i(h),0) \mid R_i^l = u_l, R_i^h\right] 
        \mid R_i^l = u_l\right]\\
        &\quad= 
        \mathbb{E}_{P_\lambda^\circ}\left[
        m_{h,0}^0(R_i^h)
        \mid R_i^l = u_l\right]
        = 
        \int m_{h,0}^0(r)\,dP_\lambda^*(r)
        = 
        \int m_h^{\mathtt{obs}}(r)\,dP_\lambda^*(r)\\
        &\quad= \lambda \sup_{g \in \mathcal F_{\mathtt{conv}}^{-}} g(u_l) + (1-\lambda) \inf_{g \in \mathcal F_{\mathtt{conc}}^{+}} g(u_l),
    \end{align*}
    where the fourth equality uses $\mathrm{supp}(P_\lambda^*)\subseteq\bar{\mathcal S}$ and $m_{h,0}^0(r)=m_h^{\mathtt{obs}}(r)$ on $\bar{\mathcal S}$.
    Because $\lambda$ ranges over $[0,1]$, every value inbetween the convex-minorant and concave-majorant bounds is attained, and consequently every point in the identified set stated in Proposition~\ref{prop:bounds} is attained by an observationally equivalent full-data distribution satisfying all maintained assumptions.
    This proves the sharpness of the identified set.
\end{proof}

\section{Additional Empirical Results}\label{sec:addemp}

This section reports the corresponding results for students who were in fifth grade in 2015 and compares them with the sixth-grade results presented in the main text.

Figure~\ref{fig_sa} plots expected subsequent test scores conditional on latent rank for both grades. The RD jump at the scholarship cutoff is positive for fifth-grade students as well, although it is smaller than that for sixth-grade students. Near the cutoff, the fifth-grade sample does not exhibit the pronounced downward shift in outcomes observed among sixth-grade students assigned to the scholarship arm.
\begin{figure}[t]
    \centering
    \begin{subfigure}[b]{0.49\textwidth}
        \centering
        \includegraphics[width=\linewidth,trim=0 10 30 50,clip]{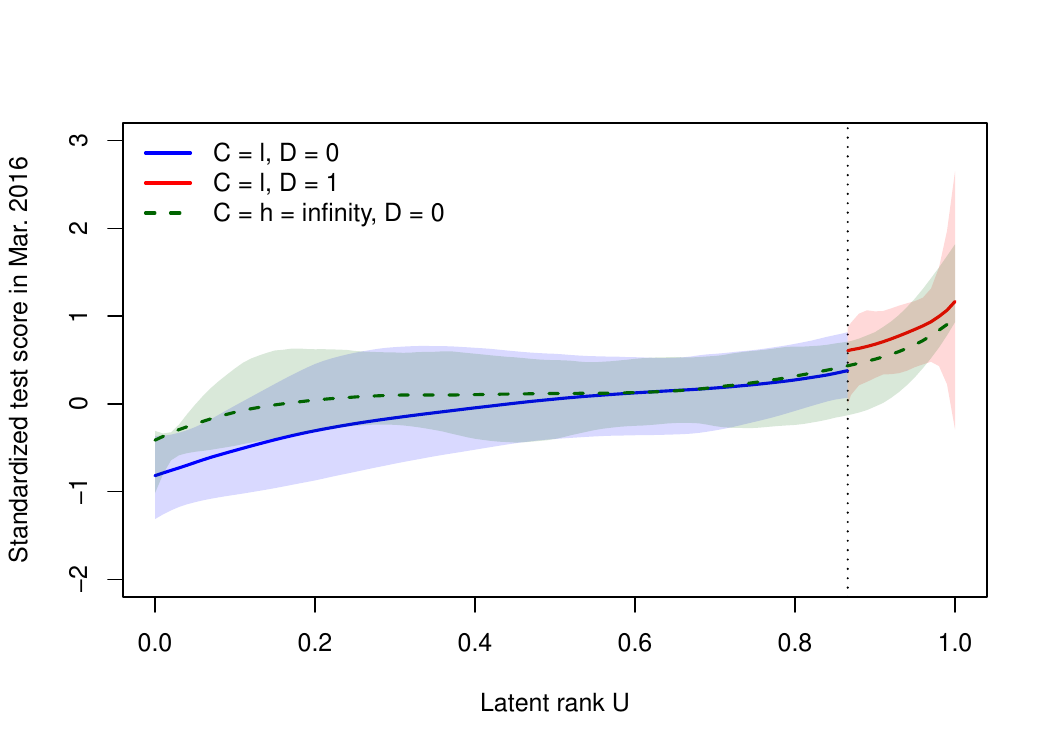}
        \caption{Grade 5}
        \label{fig5_sa}
    \end{subfigure}
    \hfill
    \begin{subfigure}[b]{0.49\textwidth}
        \centering
        \includegraphics[width=\linewidth,trim=0 10 30 50,clip]{plots/fig_grade6.pdf}
        \caption{Grade 6}
        \label{fig6_sa}
    \end{subfigure}
    \caption{Expected Test Scores Conditional on the Rank}
    \label{fig_sa}

    \begin{flushleft}
    \footnotesize
    \renewcommand{\baselineskip}{11pt}
    \textbf{Note:} The estimated IMSE-optimal bandwidths are used. The shaded regions indicate the 95\% robust-bias corrected (RBC) confidence intervals.
    The RBC standard errors are clustered at school level. The cutoffs on the latent-rank scale are approximately $0.87$ and $0.89$. The top-15-percent criterion was computed among all students, including students in the control arm and those excluded from our analysis. The corresponding cutoff in our analysis sample therefore need not equal $0.85$. 
    \end{flushleft}
\end{figure}

Table~\ref{tab:te_sa} reports the corresponding treatment-effect estimates. At the cutoff rank $u_l$, the point estimate of the conventional RD effect is positive for fifth-grade students but smaller than that for sixth-grade students.
In contrast to the sixth-grade results, the estimated incentive-response effect for fifth-grade students is small and sensitive to the bandwidth choice: it is slightly negative using the IMSE bandwidth and slightly positive using the RD bandwidth. More generally, except at rank $u=0$, the estimated incentive-response effects are substantially smaller in magnitude for fifth-grade students than for sixth-grade students.
\begin{table}[t]
    \begin{center}
    \scalebox{0.93}{
    \begin{tabular}{lccccc}\hline\hline
        \multirow{2}{*}{Rank $u$} & \multirow{2}{*}{$0.0$} & \multirow{2}{*}{$0.4$} & \multirow{2}{*}{$0.8$} & \multicolumn{2}{c}{$u_l \approx 0.87$ or $0.89$}\\
         & & & & (IMSE BW) & (RD BW)\\\hline
        \textbf{Grade 5}\\
        Treatment & -- & -- & -- & $0.227$ & $0.184$\\
        \,\,\, assignment & -- & -- & -- & [$-0.59, 0.56$] & [$-1.16, 0.72$]\\
        Incentive & $-0.404$ & $-0.151$ & $-0.043$ & $-0.056$ & $0.144$\\
        \,\,\, response & [$-0.77, 0.42$] & [$-0.78, 0.73$] & [$-0.49, 0.68$] & [$-0.41, 0.72$] & [$-0.60, 1.13$]\\
         \multirow{2}{*}{Total} & $-0.404$ & $-0.151$ & $-0.043$ & $0.172$ & $0.328$\\
         & [$-0.77, 0.42$] & [$-0.78, 0.73$] & [$-0.49, 0.68$] & [$-0.47, 0.75$] & [$-0.42, 0.50$]\\\hline
         \textbf{Grade 6}\\
        Treatment & -- & -- & -- & $0.418$ & $0.424$\\
        \,\,\, assignment & -- & -- & -- & [$-1.01, 1.73$] & [$0.60, 3.02$]\\
        Incentive & $0.308$ & $-0.391$ & $-0.782$ & $-1.073$ & $-1.109$\\
        \,\,\, response & [$-0.26, 1.41$] & [$-0.97, -0.05$] & [$-1.27, -0.08$] & [$-1.63, -0.05$] & [$-2.67, -1.32$]\\
         \multirow{2}{*}{Total} & $0.308$ & $-0.391$ & $-0.782$ & $-0.655$ & $-0.685$\\
         & [$-0.26, 1.41$] & [$-0.97, -0.05$] & [$-1.27, -0.08$] & [$-1.67, 0.71$] & [$-1.26, 0.89$]\\\hline
    \end{tabular}}
    \caption{Treatment Effects}
    \label{tab:te_sa}
    \end{center}

    \footnotesize
    \renewcommand{\baselineskip}{11pt}
    \textbf{Note:} 
    The conventional point estimates and 95\% robust bias-corrected confidence intervals are reported. Because the confidence intervals incorporate bias correction, they are not necessarily centered on the corresponding point estimates.
    For all columns except the last, we use the estimated IMSE bandwidths used in Figure~\ref{fig}. In the last column (RD BW), the MSE-optimal bandwidths for RD estimator proposed in \cite{Imbens_Kalyanaraman:2011} are used for all estimators.
    The robust-bias corrected standard error is clustered at school level. 
\end{table}

One possible interpretation is that the threshold-based incentive generated less discouragement among fifth-grade students, except perhaps among those at the very bottom of the rank distribution. Consistent with this possibility, Figure~\ref{qte_sa} shows that the negative quantile treatment effects of threshold assignment on the June 2015 test score are much less pronounced for fifth-grade students than for sixth-grade students, except at very low ranks. Thus, the threshold policy appears to have had a considerably weaker negative effect on the contemporaneous score distribution in fifth grade, which may help explain why the adverse incentive-response effect on subsequent achievement is also much smaller.
\begin{figure}[t]
    \centering
    \begin{subfigure}[b]{0.49\textwidth}
        \centering
        \includegraphics[width=\linewidth,trim=0 10 30 50,clip]{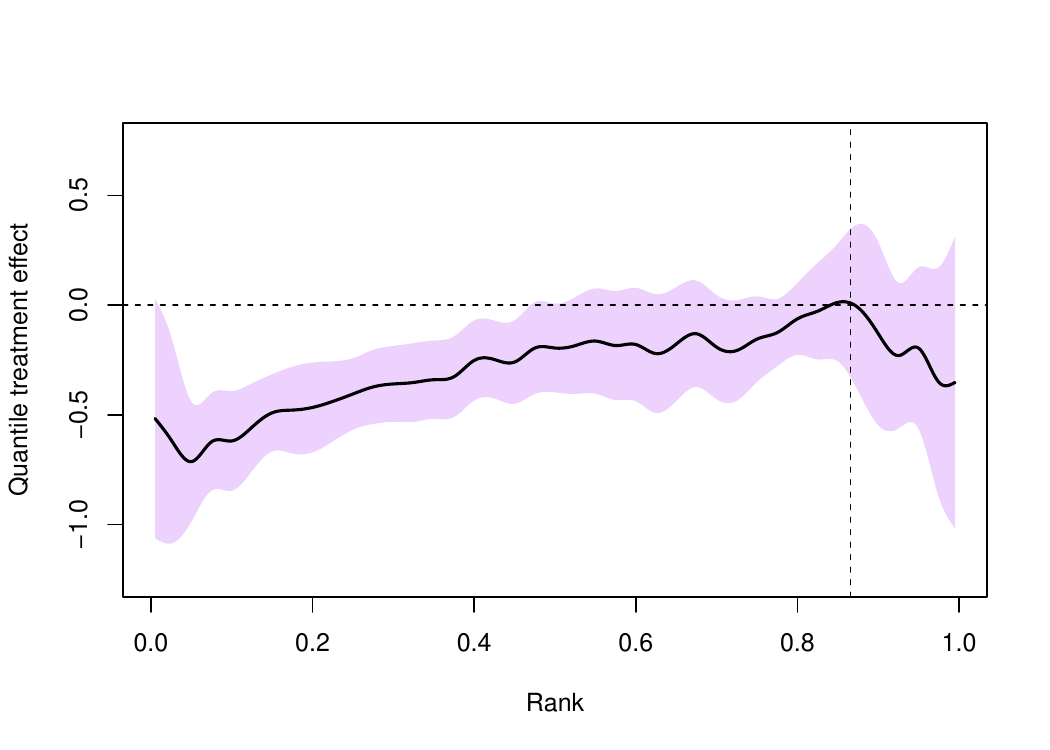}
        \caption{Grade 5}
        \label{fig:qte5_sa}
    \end{subfigure}
    \hfill
    \begin{subfigure}[b]{0.49\textwidth}
        \centering
        \includegraphics[width=\linewidth,trim=0 10 30 50,clip]{plots/qte_grade6.pdf}
        \caption{Grade 6}
        \label{fig:qte6_sa}
    \end{subfigure}
    \caption{Quantile Treatment Effect of Threshold Assignment on Test Score in June 2015}
    \label{qte_sa}
\end{figure}

These grade-specific patterns also illustrate how the policy implications can differ even when the conventional RD effect is positive in both groups. For sixth-grade students, the positive treatment-assignment effect is accompanied by a sizable negative incentive-response effect, raising the possibility that implementing the scholarship through a merit-based threshold may be undesirable despite a positive effect of scholarship receipt itself. For fifth-grade students, by contrast, the negative incentive-response effect appears much smaller relative to the treatment-assignment effect, suggesting that merit-based implementation may be more reasonable in this group. At the same time, among the fifth-grade students, the incentive-response effect appears to become more negative at lower ranks. This pattern suggests that the preferred policy rule may depend on students' baseline academic achievement, with less demanding or alternative assignment rules potentially being more appropriate for lower-performing students.

\end{document}